\documentclass[journal]{IEEEtran}
\usepackage{amssymb}
\usepackage{mathrsfs}
\usepackage{graphicx}
\usepackage{amsmath, amssymb, amsthm}
\usepackage{leftidx}
\usepackage{extarrows}
\usepackage{stfloats}
\usepackage{picinpar}
\usepackage{enumerate}
\usepackage{algorithm}
\usepackage{algorithmicx}
\usepackage{algpseudocode}
\usepackage{epsfig}
\usepackage{latexsym}
\usepackage{mathrsfs}
\usepackage{amsfonts}
\usepackage{enumerate}
\usepackage{graphics}
\usepackage{graphicx,subfigure}
\usepackage{MnSymbol}
\usepackage{float}
\usepackage{pict2e}
\usepackage{tikz}
\usepackage{bm}
\newtheorem{theorem}{Theorem}

\newtheorem{lemma}{Lemma}
\newtheorem{remark}{Remark}
\newtheorem{definition}{Definition}

\title{Distributed Pinning Set Stabilization of Large-Scale Boolean Networks}

\author{Shiyong~Zhu, Jianquan Lu, \IEEEmembership{Senior Member, IEEE}, Liangjie Sun, and Jinde Cao, \IEEEmembership{Fellow, IEEE} % <-this % stops a space
\thanks{Corresponding author: Jianquan Lu}
\thanks{Shiyong Zhu and Jianquan Lu are with the Department of Systems Science, School of Mathematics, Southeast University, Nanjing 210096, China (email: zhusy0904@gmail.com; jqluma@seu.edu.cn).}
\thanks{Liangjie Sun is with the Advanced Modeling and Applied Computing Laboratory, Department of Mathematics, The University of Hong Kong, Hong Kong, China.}
\thanks{Jinde Cao is with the School of Mathematics, Frontiers Science Center for Mobile Information Communication and Security, Southeast University, Nanjing 210096, China, and also with the Purple Mountain Laboratories, Nanjing 211111, China, and also with the Yonsei Frontier Lab, Yonsei University, Seoul 03722, South Korea (e-mail: jdcao@seu.edu.cn).}}

\begin{document}
\maketitle
\thispagestyle{empty}
\pagestyle{empty}
%%%%%%%%%%%%%%%%%%%%%%%%%%%%%%%%%%%%%%%%%%%%%%%%%%%%%%%%%%%%%%%%%%%%%%%%%%%%%%%%%%%%%%%%%%%%%%%%%%%%%%%%%%%%%%%%%%%%%%%%%%%%%%%%%%%%%%%%%%%%%%%%%%%%%%%%%%%
\begin{abstract}
  In this article, we design the distributed pinning controllers to globally stabilize a Boolean network (BN), specially a sparsely connected large-scale one, towards a preassigned subset of state space \textcolor[rgb]{1,0,0}{through} the node-to-node message exchange. Given an \textcolor[rgb]{1,0,0}{appointed} state set, system nodes are partitioned into two \textcolor[rgb]{1,0,0}{disjoint} parts, \textcolor[rgb]{1,0,0}{which respectively gather the nodes whose states are fixed or arbitrary with respect to the given state set.} \textcolor[rgb]{1,0,0}{With such node division, three parts of pinned nodes are selected and the state feedback controllers are accordingly designed such that the resulting BN satisfies three conditions: the states of the other nodes cannot affect the nodal dynamics of fixed-state nodes, the subgraph of network structure induced by the fixed-state nodes is acyclic, and the steady state of the subnetwork induced by the fixed-state nodes lies in the state set given beforehand. If the BN after control is acyclic, the stabilizing time is revealed to be no more than the length of the longest path in the current network structure plus one. This enables us to further design the pinning controllers with the constraint of stabilizing time.} Noting that the overall procedure runs in an exponentially increasing time with respect to the largest number of functional variables in the dynamics of pinned nodes, the sparsely-connected large-scale BNs can be well addressed in a reasonable amount of time. Finally, we demonstrate the applications of our theoretical results in a T-LGL survival signal network with $29$ nodes \textcolor[rgb]{1,0,0}{and T-cell receptor signaling network with $90$ nodes}.
\end{abstract}


%%%%%%%%%%%%%%%%%%%%%%%%%%%%%%%%%%%%%%%%%%%%%%%%%%%%%%%%%%%%%%%%%%%%%%%%%%%%%%%%%%%%%%%%
\section{Introduction}\label{section-introduction}
Since Kauffman proposed a logically coupled model with binary states--Boolean network (BN)--to describe the qualitative state evolution of gene regulatory networks in $1969$ (cf. \cite{kauffman1969jtb437}), it has received a surge of interest from the scientists in an amount of areas, such as systems biology \cite{akutsu2007jtb670}, multi-agent cooperation \cite{fagiolini2013auto2339}, power engineering \cite{power_systems}, as well as game theory \cite{chengdz2018auto51}. \textcolor[rgb]{1,0,0}{Formally, the nodal dynamics of an $n$ nodes' BN can be described as follows:}
\begin{equation}\label{equ-BN}
{\bm x}_j(t+1)={\bm f}_j([{\bm x}_i(t)]_{i\in \mathbf{N}_j}),~j=1,2,\cdots,n,
\end{equation}
where $t=1,2,\cdots$ are the discrete time instants; ${\bm x}_j\in \mathscr{D}:=\{1,0\}$ is the state variable of node $j$; $\mathbf{N}_j$ is the \textcolor[rgb]{1,0,0}{list}\footnote{\textcolor[rgb]{1,0,0}{In this article, we term a set with elements ranking in an increasing order as a list.}} of the indices of all functional variables of ${\bm f}_j$; \textcolor[rgb]{1,0,0}{and abbreviation $[{\bm x}_i(t)]_{i\in \mathbf{N}_j}$ is the queue of variables ${\bm x}_i(t)$ with subscripts belonging to the list $\mathbf{N}_j$ and ordering the same as in $\mathbf{N}_j$.}

About a decade ago, a systematic theoretical framework--algebraic state space representation (ASSR) \textcolor[rgb]{1,0,0}{approach}--for BNs was proposed by Cheng {\em et al.} based on the semi-tensor product (STP) of matrices \textcolor[rgb]{1,0,0}{(cf. monograph \cite{chengdz2011springer} and the references therein). Rewriting} each state \textcolor[rgb]{1,0,0}{variable} in its canonical form, the $n$-dimensional system state ${\bm x}(t):=({\bm x}_1(t),{\bm x}_2(t),\cdots, {\bm x}_n(t))$ can be one-to-one represented as a canonical vector $x(t)$ \textcolor[rgb]{1,0,0}{of order $2^n$. With this representation and the properties of the STP of matrices,} BN (\ref{equ-BN}) can be equivalently expressed \textcolor[rgb]{1,0,0}{into the following standard linear time-invariant system, which is referred to as the ASSR of BN (\ref{equ-BN}):}
\begin{equation}\label{equ-assr}
x(t+1)=Lx(t)
\end{equation}
where logical matrix $L$ of size $2^n\times 2^n$ is called the network transition matrix of BN (\ref{equ-BN}). \textcolor[rgb]{1,0,0}{To date, expression (\ref{equ-assr}) facilitates the emergence of} many remarkable results including but not limited to controllability \cite{margaliot2012aut1218,liuy2015auto340}, observability \cite{Valcher2012TAC1390,Margaliot2013Observability2351,zhusy2021tac}, stabilization \cite{valcher2015aut21,lujq2018siam4385,chenhw2020tac,zhusy2019tac}, synchronization \cite{lir2012tnnls840}, and robustness analysis \cite{liht2020siam3632,liht2021tac1231}. \textcolor[rgb]{1,0,0}{However, noting that the size of matrix $L$ is $2^n\times 2^n$, the time complexity of the theoretical results \textcolor[rgb]{1,0,0}{built on ASSR approach} will be $O(2^n)$ at least except that some complexity reduction techniques are adopted beforehand like network aggregation. As a result, the $O(2^n)$-barrier limits the applicability of ASSR approach to the large-scale BNs in practice.} Taking the controllability of BNs as an example, the STP toolbox provided by Cheng {\em et al.} can only deal with the controllability of BNs with nodes $n\leq25$ or so in a reasonable amount of time if utilizing the input-state transition matrix derived in \cite{zhaoy2010scl767}.

\textcolor[rgb]{1,0,0}{In this article, we consider the controllers design problem for the set stabilization--which was proposed by Guo {\em et al.} \cite{guoyq2015auto106} at first in the area of BNs--of large-scale BNs when the target domain is a subset of state space given beforehand rather than a single state. One motivation of set stabilization is the biological observation that only a part of essential nodes in almost networks should be taken care specially for large-scale BNs, in which case the considered BN is only required to be globally stabilized at a state set.} Moreover, several frequent problems of BNs can be thought of as \textcolor[rgb]{1,0,0}{the variations} of set stabilization, including synchronization \cite{lir2012tnnls840}, output tracking \cite{liht2015aut54}, and output regulation \cite{liht2017tac2993}. However, \textcolor[rgb]{1,0,0}{to our best knowledge, the analysis and controller synthesis on the set stabilization of BNs and its variations also have the $O(2^n)$-barrier of the time complexity.}

\textcolor[rgb]{1,0,0}{In recent years, the pinning control design has attracted much attention in the field of BNs (see, e.g., \cite{lujq2016ieeetac1658,Liff2016TNNLS1585,liff2018tnnls,liht2017jfi3039,liurj2017neuro142,zhongjie2019new,zhusy2020framework,zhusy2021controlgraph,zhusy2021sensor})}
and has been recognized because the whole dynamical behaviors of real networks with a mass of nodes can always be manipulated by controlling a small fraction of nodes \cite{liuyy2011nature167}. An illustrative assertion is the worm {\em C. elegans} with $297$ nerve cells, whose entire body can be provoked by average $49$ nodes occupying only about $17\%$ on total \cite{cho2011scientific}. \textcolor[rgb]{1,0,0}{In the area of BNs, injecting the control inputs on a preassigned part of nodes, the concept of pinning controllability appeared in \cite{lujq2016ieeetac1658} for the first time without consideration of the specific control form and the design procedure. A more concrete version of a pinning controlled BN, which was written in the following form, was provided in \cite{Liff2016TNNLS1585}:}
\begin{equation}\label{equ-BN-pin}
{\bm x}_j(t+1)=\left\{\begin{array}{l}
{\bm u}_j(t) \oplus_j {\bm f}_j({\bm x}(t)),~j\in\psi\\
{\bm f}_j({\bm x}(t)),~j\not\in\psi
\end{array}\right.
\end{equation}
where $\psi$ is the pinned node set; $\oplus_j$ is the logical operator that couples the control input and the original dynamics of node $j$; and input ${\bm u}_j(t)$ can be either open-loop control or state feedback control ${\bm u}_j(t)=\phi_j({\bm x}(t))$. \textcolor[rgb]{1,0,0}{In (\ref{equ-BN-pin}),} three configurations--pinned node set $\psi$, state feedback gain $\phi_j$ and logical coupling $\oplus_j$--\textcolor[rgb]{1,0,0}{are unknown and wait to be designed. For this purpose, changing certain columns of matrix $L$ in the noncontrolled BN (\ref{equ-assr}), the desired transition matrix $L'$ can be derived to achieve the anticipated performance. Finally, while pinning the nodes whose structure matrices decomposed from matrix $L'$ are different from the original ones,} the unknown components $\oplus_j$ and $\phi_j$ can be calculated by solving their structure matrices from a series of \textcolor[rgb]{1,0,0}{logical} matrix equations. Except for the global stabilization studied in \cite{Liff2016TNNLS1585}, such type of pinning controllers has been similarly applied to the set stabilization of BNs \cite{liurj2017neuro142} and probabilistic BNs \cite{liff2018tnnls}, and its variations \cite{liht2017jfi3039}. Hereafter, this type of design scheme for pinning controllers is termed the {\em $L$-based pinning approach}, for which, \textcolor[rgb]{1,0,0}{other than the exponential time complexity}, the obtained state feedback gains are always functional with respect to the majority of state variables such that the designed control form is a bit complicated.

\textcolor[rgb]{1,0,0}{Inspired by all aforementioned motivations, we shall be developing an efficient pinning control strategy that is capable of globally stabilizing large-scale BNs to a given state set efficiently. Our work is closest to the recent papers \cite{zhongjie2019new,zhusy2020framework,zhusy2021controlgraph,zhusy2021sensor}, all of which are dedicated to analyzing or controlling the global dynamical behaviors of large-scale BNs by using the ($n \times n$)-dimensional network structures even without available knowledge of nodal dynamics. In \cite{zhongjie2019new}, Zhong {\em et al.} have proposed a {\em distributed} pinning control approach, which is the prototype of the controllers we study here, for the global stabilization of BNs by resorting to a lemma disclosing that the acyclic network structure suffices to the global stability of BNs. By characterizing} the structural controllability of BNs without knowing their nodal dynamics, distributed pinning controllers have also been designed in \cite{zhusy2020framework} to make the considered BNs be controllable. Moreover, the structural observability of BNs has also been formalized in \cite{zhusy2021controlgraph}, and resulted in a new approach to design the sensors from the perspective of pinning observability (cf. \cite{zhusy2021controlgraph}, \cite{zhusy2021sensor}).

\textcolor[rgb]{1,0,0}{Despite the distributed pinning controllers succeed in overcoming the aforementioned limitations of $L$-based pinning approach, the design of such controllers relys on the network-structure-based criteria for the target dynamical behaviors (cf. e.g., \cite{zhongjie2019new,zhusy2020framework,zhusy2021controlgraph,zhusy2021sensor}). Thus, all the works cited above cannot deal with the set stabilization of BNs on account of the lack of the sufficient network-structure-based criterion for set stability. This motivates us to study the design of distributed pinning controllers for the set stabilization of BNs. We conclude the following three differences of this work in comparison with the $L$-based pinning approach to set stabilization (see, e.g., \cite{liff2018tnnls,liht2017jfi3039,liurj2017neuro142}) and the distributed pinning controllers for global stabilization \cite{zhongjie2019new}:}
\begin{itemize}
  \item[1)] Given a target state set beforehand, all the network nodes are partitioned into two disjoint node sets in accordance with their states with respect to the given set: the one gathers the nodes whose states can be arbitrary, and the other is a collection of fixed-state nodes. An additional part of nodes is pinned to cut off the data flow from the other nodes to the fixed-state nodes. \textcolor[rgb]{1,0,0}{Both of them are optional in \cite{zhongjie2019new}.}

  \item[2)] \textcolor[rgb]{1,0,0}{We reveal that, for any BN with acyclic network structure, the stabilizing time is no more than the length of the longest path in the network structure plus one. As a consequence, the distributed pinning controllers are designed with the constraint of stabilizing time; this} is not considered in \cite{zhongjie2019new}. Moreover, we first finish searching all the pinned nodes and then proceed to design the state feedback controllers and logical couplings; this results in a more concise form in controllers.

  \item[3)] Compared with the $L$-based pinning approach (see, e.g., \cite{liff2018tnnls,liht2017jfi3039,liurj2017neuro142}), the controllers designed here preserve all the advantages of the distributed pinning controllers, for which the application of network structures makes the time complexity to be $O(n^2+ms^2+s2^{\epsilon})$ and breaks the $O(2^n)$ barrier, where $s$ is the number of pinned nodes, $\epsilon$ is the maximum in-degree of pinned nodes and $m$ is the sum of number of functional variables for all fixed-state nodes. Some empirical observations for the sparse connection of biological network (cf., e.g., \cite{jeong2000nature}) facilitates the applicability of such controllers to large-scale BNs because of $\epsilon \ll n$. Moreover, all the pinned nodes can be determined in polynomial time. Last but not least, state feedback gain injected on each pinned node only relies on the in-neighbors' message rather than the global information.
\end{itemize}

The remainder of this article is arranged as follows. We start in Section \ref{section-preliminaries} by introducing some preliminaries including brief review of the STP of matrices and some basic notions and lemmas. Section \ref{section-newmethod} elaborates the design of the distributed pinning controllers. Two biological examples are considered in Section \ref{section-example}, while Section \ref{section-conclusion} concludes this article.

\textcolor[rgb]{1,0,0}{{\em Notations:} Let $\mathbb{N}$ and $\mathbb{R}$ be the sets of integers and real numbers, respectively. Given $m,n\in\mathbb{N}$ with $m<n$, set $\{m,m+1,\ldots,n\}$ is abbreviated as $[m,n]_{\mathbb{N}}$. Defining by $\delta_n^k$ the $k$th column of identity matrix $I_n$, set $\Delta_n$ is composed of all canonical vectors $\delta_n^k$ with $k\in[1,n]_{\mathbb{N}}$. We call $B:=(b_{ij})_{n\times n}$ a Boolean matrix of size $n \times n$ if $b_{ij}\in\mathscr{D}$. $\mathscr{L}_{m\times n}$ is defined by the set of all logical matrices of order $m\times n$, which can be written as $[\delta_m^{i_1},\delta_m^{i_2},\ldots,\delta_m^{i_n}]$ and abbreviated as $\delta_n[i_1,i_2,\cdots,i_n]$. Given a set $C$, $\mid C \mid$ denotes its cardinality. $+_{\mathscr{B}}$ and $\times_{\mathscr{B}}$ are Boolean addition and Boolean product, respectively.}

\section{Preliminaries}\label{section-preliminaries}
\subsection{STP of Matrices}
In this subsection, we briefly review the STP of matrices and its properties. For survey, we refer the readers to monograph \cite{chengdz2011springer} and the references therein.
\begin{definition}[See \cite{chengdz2011springer}]
Given a ($p\times q$)-dimensional matrix $U$ and an ($s\times t$)-dimensional matrix $V$, their STP is defined as $U\ltimes V = (U\otimes I_{\frac{\beta}{q}})(V\otimes I_{\frac{\beta}{s}})$, where $\beta$ is the least common multiple of $q$ and $s$.
\end{definition}

\textcolor[rgb]{1,0,0}{The STP of matrices becomes the traditional matrix product verbatim when $q=s$, so it extends the traditional matrix product to the product of non-matching-dimensions matrices. Such extension rationalizes the following properties of the STP of matrices.}
\begin{lemma}[See \cite{chengdz2011springer}]\label{lem-stp}
Given real matrices ${\bf a}$, ${\bf b}$ and $C$ of sizes $p \times 1$, $q \times 1$ and $s \times t$, respectively, it holds that
\begin{itemize}
  \item[1)] ${\bf a} \ltimes C=(I_p \otimes C)\ltimes {\bf a}$;
  \item[2)] ${\bf a}\ltimes {\bf b}=W_{[q,p]} \ltimes {\bf b} \ltimes {\bf a}$, where $W_{[q,p]}:=[I_p\otimes\delta_q^1,I_p\otimes\delta_q^2, \ldots, I_p\otimes\delta_q^q]$ is the ($qp\times qp$)-dimensional swap matrix;
  \item[3)] ${\bf a} \ltimes {\bf a}=\Phi_p \ltimes {\bf a}$, where $\Phi_p:=[\delta_p^1\ltimes \delta_p^1, \delta_p^2\ltimes\delta_p^2,\ldots,\delta_p^p\ltimes\delta_p^p]$ is the power-reducing matrix for the $p$-ordered column vector.
\end{itemize}
\end{lemma}

Defining the canonical form of variable ${\bm x}_i$ by $\zeta(\textcolor[rgb]{1,0,0}{{\bm x}_i}):=\delta_2^{2-{\bm x}_i}$, any logical function can be rewritten in its multilinear form.
\begin{lemma}[See \cite{chengdz2011springer}]\label{lem-strMatrix}
Given a logical function ${\bm f}({\bm x}_1,{\bm x}_2,\cdots,{\bm x}_n):\mathscr{D}^n\rightarrow\mathscr{D}$, there exists a unique logical matrix $S_{{\bm f}}\in\mathscr{L}_{2\times 2^n}$ such that $\zeta({\bm f}({\bm x}_1,{\bm x}_2,\cdots,{\bm x}_n))=S_{{\bm f}}\ltimes x_1 \ltimes x_2 \ltimes\cdots \ltimes x_n$, where $x_i:=\zeta(\textcolor[rgb]{1,0,0}{{\bm x}_i})$, $i\in[1,n]_{\mathbb{N}}$ and matrix $S_{{\bm f}}$ is called the structure matrix of function ${\bm f}$.
\end{lemma}

\subsection{Basic Notions and Lemmas}
\begin{definition}\label{def-functional}
For a logical function ${\bm f}({\bm x}_1,{\bm x}_2,\ldots,{\bm x}_n):\mathscr{D}^n\rightarrow\mathscr{D}$, variable ${\bm x}_j$ therein is said to be functional if there exists a vector $(\eta_1,\eta_2,\ldots,\eta_n)\in\mathscr{D}^n$ satisfying ${\bm f}(\eta_1,\ldots,\eta_{j-1},\eta_j,\eta_{j+1},\ldots,\eta_n) \neq {\bm f}(\eta_1,\cdots,\eta_{j-1},\overline{\eta_j},\eta_{j+1},\ldots,\eta_n)$.
\end{definition}

Given an integer $n\in\mathbb{N}$, for any index $i\in[1,n]_{\mathbb{N}}$, subset $\vec{\mathscr{L}}^i_{2\times 2^{n}}\subseteq\mathscr{L}_{2\times 2^n}$ consists of all the matrices $L\in \vec{\mathscr{L}}^i_{2\times 2^n}$ satisfying
$$LW_{[2,2^{i-1}]}\ltimes\delta_2^1 +_{\mathscr{B}} LW_{[2,2^{i-1}]}\ltimes\delta_2^2\in\mathscr{L}_{2 \times 2^{n-1}},$$
based on which, we define $\vec{\mathscr{L}}_{2\times 2^n}:=\{\vec{\mathscr{L}}^1_{2\times 2^{n}},\vec{\mathscr{L}}^2_{2\times 2^{n}},\ldots,\vec{\mathscr{L}}^n_{2\times 2^{n}}\}$ and $\vec{\mathscr{L}}^c_{2\times 2^n}:=\mathscr{L}_{2\times 2^n}\backslash\vec{\mathscr{L}}_{2\times 2^n}$. Then, we have the following lemma for logical function ${\bm f}({\bm x}_1,{\bm x}_2,\ldots,{\bm x}_n)$ whose functional variables are ${\bm x}_{\varsigma_1},{\bm x}_{\varsigma_2},\ldots,{\bm x}_{\varsigma_{\sigma}}$ with $1\leq\varsigma_{1}<\varsigma_{2}<\cdots<\varsigma_{\sigma}\leq n$.
\begin{lemma}\label{lem-extract}
\textcolor[rgb]{1,0,0}{Given a logical function ${\bm f}({\bm x}_1,{\bm x}_2,\ldots,{\bm x}_n)$ with all functional variables being ${\bm x}_{\varsigma_1},{\bm x}_{\varsigma_2},\ldots,{\bm x}_{\varsigma_{\sigma}}$, there must exist a matrix $A\in\vec{\mathscr{L}}^c_{2\times 2^{n}}$ such that $S_{{\bm f}}W=A(I_{2^{\sigma}}\otimes {\bf 1}^\top_{2^{n-\sigma}})$ holds for the structure matrix $S_{\bm f}$ of function ${\bm f}$, and vice versa, where $W:=W_{[2,2^{\varsigma_{\sigma}-1}]} \ltimes W_{[2,2^{\varsigma_{\sigma-1}}]} \ltimes W_{[2,2^{\varsigma_{\sigma-2}+1}]} \ltimes \cdots \ltimes W_{[2,2^{\varsigma_{1}+\sigma-2}]}$.}
\end{lemma}
\begin{proof}
\textcolor[rgb]{1,0,0}{By Lemma \ref{lem-stp}, it follows that
\begin{equation}\label{equ-lemma2.3-1}
\begin{aligned}
\zeta({\bm f}({\bm x}_1,{\bm x}_2,\ldots,{\bm x}_n))&=S_{{\bm f}}\ltimes_{i=1}^{n}x_i\\
&=S_{{\bm f}}W(\ltimes_{i=1}^{\sigma}x_{\varsigma_i})(\ltimes_{i\in[1,n]_{\mathbb{N}}\backslash[1,\sigma]_{\mathbb{N}}}x_i(t)).
\end{aligned}
\end{equation}}

\textcolor[rgb]{1,0,0}{Except for variables ${\bm x}_{\varsigma_1},{\bm x}_{\varsigma_2},\ldots,{\bm x}_{\varsigma_{\sigma}}$, the other variables are all nonfunctional. Thus, matrix $S_{\bm f}W$ must can be written as $S_{{\bm f}}W=A\otimes {\bf 1}^\top_{2^{n-\sigma}}$. Then, we have that
$$ \zeta({\bm f}({\bm x}_1,{\bm x}_2,\ldots,{\bm x}_n))=A(\ltimes_{i=1}^{\sigma}x_{\varsigma_i}). $$
For matrix $A$ therein, it does belong to set $\vec{\mathscr{L}}_{2 \times 2^n}$. Otherwise, some variables in ${\bm x}_{\varsigma_1},{\bm x}_{\varsigma_2},\ldots,{\bm x}_{\varsigma_{\sigma}}$ are non-functional. Thus, one has that}
$$ S_{{\bm f}}W= A\otimes {\bf 1}^\top_{2^{n-\sigma}} = A(I_{2^\sigma} \otimes {\bf 1}^\top_{2^{n-\sigma}}).$$

The converse conclusion is obvious.
\end{proof}

Given BN (\ref{equ-BN}), we write its incidence matrix as $I({\bm f}):=(I_{ij})_{n\times n}$, where $I_{ij}=1$ if ${\bm x}_j$ is a functional variable of ${\bm f}_i$, and  $I_{ij}=0$, otherwise. Since an incidence matrix is a Boolean one, it can be regarded as the adjacency matrix of the {\em network structure} of BN (\ref{equ-BN}), which is a digraph ${\bf G}:=({\bf V},{\bf E})$ with the vertex set ${\bf V}:=\{v_1,v_2,\cdots,v_n\}$ and the edge set ${\bf E}$ composed of arcs $e_{ij}$--which joint vertex $v_i$ to vertex $v_j$--for all $i$ and $j$ with $I_{ji}=1$. For every arc $e_{ij}\in\mathbf{E}$, its ending and starting vertices are respectively denoted by ${\bf O}^{+}_{e_{ij}}=\{v_j\}$ and ${\bf O}^{-}_{e_{ij}}=\{v_i\}$. \textcolor[rgb]{1,0,0}{Given vertices $v_i,v_j\in{\bf V}$, a path from $v_i$ to $v_j$ is a sequence $v_{k_0}v_{k_1}\ldots v_{k_q}$ with $v_{k_0}=v_i$, $v_{k_q}=v_j$ and $e_{k_ik_{i+1}}\in{\bf E}$. Specially, if $v_i=v_j$, it is called a cycle. Hereafter,} \textcolor[rgb]{1,0,0}{in order to distinguish a node of BN (\ref{equ-BN}) and its corresponding vertex in the network structure, we term them a {\em node} and a {\em vertex}}, respectively.

\begin{definition}[See \cite{guoyq2015auto106}]\label{def-setstability}
Given a state set ${\bm \Lambda}\subseteq\mathscr{D}^n$, BN (\ref{equ-BN-pin}) is said to be globally ${\bm \Lambda}$-stabilized by pinning control if starting from any initial state ${\bm x}_0:={\bm x}(0)\in\mathscr{D}^n$, there exists the minimal integer $\tau^\ast_{{\bm x}_0}\in\mathbb{N}$ depending on ${\bm x}_0$ such that ${\bm x}(t)\in{\bm \Lambda}$ holds for any $t\geq \tau^\ast_{{\bm x}_0}$. \textcolor[rgb]{1,0,0}{Furthermore,} number $\tau^\ast:=\max\{\tau^\ast_{{\bm x}_0} \mid {\bm x}_0\in\mathscr{D}^n\}$ is called the stabilizing time of BN (\ref{equ-BN-pin}) with respect to the set ${\bm \Lambda}$.
\end{definition}
\begin{remark}
Specially, if ${\bm \Lambda}$ only contains single state $\alpha$, BN (\ref{equ-BN-pin}) is said to be globally stabilized at $\alpha$ by pinning control.
\end{remark}

Then, a network-structure-based stability criterion of BN (\ref{equ-BN}) is presented. Although Lemma \ref{lemma-globalconvergence} is only sufficient, it suffices to design the pinning controllers for BN (\ref{equ-BN}).
\begin{lemma}[See \cite{robert2012discrete}]\label{lemma-globalconvergence}
BN (\ref{equ-BN-pin}) is globally stabilized at a certain state if its network structure is acyclic.
\end{lemma}

Finally, we present a stronger result that is adapted from Lemma \ref{lemma-globalconvergence} with the help of the following lemma. Define the diameter of an acyclic digraph, denoted by$\text{diam}(\cdot)$, by the length of the longest path in this digraph. The Hamming distance of any two vectors $\mu,\nu\in\mathscr{D}^n$ is defined as $\text{dist}(\mu,\nu):=\Sigma_{i=1}^{n} (\mu_i \bar{\vee} \nu_i)$.
\begin{lemma}[See \cite{robert2012discrete}]
For BN (\ref{equ-BN}), it holds that
\begin{equation}\label{equ-discrete-differ}
  \text{dist}({\bm f}(\mu),{\bm f}(\nu))\leq I_{\bm f} \times_{\mathscr{B}} \text{dist}(\mu,\nu).
\end{equation}
\end{lemma}

\begin{lemma}\label{lem-within-time}
If the network structure ${\bf G}$ of BN (\ref{equ-BN-pin}) is acyclic and satisfies $\text{diam}({\bf G})=\tau$, it will be globally stabilized within time $\tau$.
\end{lemma}
\begin{proof}
Lemma \ref{lemma-globalconvergence} indicates that BN (\ref{equ-BN-pin}) with acyclic network structure is globally stabilized at certain state. Without loss of generality, denoting this steady state by $\alpha\in\mathscr{D}^n$ and following (\ref{equ-discrete-differ}), one has that
\begin{equation*}\label{equ-inequality}
\text{dist}(\alpha,{\bm f}^{\tau+1}(\mu))=\text{dist}({\bm f}^{\tau+1}(\alpha),{\bm f}^{\tau+1}(\mu)) \leq I^{\tau+1}_{\bm f} \times_{\mathscr{B}} \text{dist}(\alpha,\mu).
\end{equation*}

Due to $\text{diam}({\bf G})=\tau$, $\tau+1$ is the minimal number $k$ such that $I^{k}_{\bm f}={\bf 0}_{n \times n}$, which further implies
$$\text{dist}(\alpha,{\bm f}^{\tau+1}(\mu)) \leq I^{\tau+1}_{\bm f} \times_{\mathscr{B}} \text{dist}(\alpha,\mu)=0,~\forall \mu\in\mathscr{D}^n,$$
that is, ${\bm x}(\tau)={\bm f}^{\tau}(\mu)=\alpha$. This completes the proof of this lemma.
\end{proof}

\section{Designing Distributed Pinning Controlers}\label{section-newmethod}
In this section, we shall design the distributed pinning controllers, which are originally introduced in \cite{zhongjie2019new} for the global stabilization and are further adapted in \cite{zhusy2020framework,zhusy2021controlgraph,zhusy2021sensor} for controllability and observability, to achieve the set stabilization of BN (\ref{equ-BN}). \textcolor[rgb]{1,0,0}{Here, we mainly extend the results in \cite{zhongjie2019new} to the set stabilizing controllers of BNs, but the controller form is more concise even if the problem is more general. Besides, the stabilizing time is also considered in this article.}

\subsection{Partition of System Nodes}\label{subsec-setpartition}
For a set ${\bm \Lambda}\subseteq\mathscr{D}^n$ given beforehand and for each node $k\in[1,n]_{\mathbb{N}}$, we define sets ${\bm \Lambda}_k^0=\left\{ {\bm x}\in {\bm \Lambda} \mid {\bm x}_k=0 \right\}$ and ${\bm \Lambda}_k^1=\left\{ {\bm x}\in {\bm \Lambda} \mid {\bm x}_k=1 \right\}$. \textcolor[rgb]{1,0,0}{Deleting the $k$th component of every state in sets ${\bm \Lambda}^0_k$ and ${\bm \Lambda}^1_k$ to obtain the sets $\tilde{{\bm \Lambda}}^0_k$ and $\tilde{{\bm \Lambda}}^1_k$, respectively, if $\tilde{{\bm \Lambda}}^0_k=\tilde{{\bm \Lambda}}^1_k$, the state of node $k$ is said to be arbitrary with respect to set ${\bm \Lambda}$; otherwise, it is said to be fixed.} Formally, subset $\Xi^{\textcolor[rgb]{1,0,0}{uf}}_{\bm \Lambda}$ can be expressed as
\begin{equation}
\Xi^{\textcolor[rgb]{1,0,0}{uf}}_{\bm \Lambda}:=\{ k\in\textcolor[rgb]{1,0,0}{[1,n]_{\mathbb{N}}} \mid \tilde{{\bm \Lambda}}_k^0 = \tilde{{\bm \Lambda}}_k^1\},
\end{equation}
with complementary set $\Xi^{f}_{\bm \Lambda}:=[1,n]_{\mathbb{N}}\backslash \Xi^{uf}_{\bm \Lambda}$ called the fixed-state nodes with respect to set ${\bm \Lambda}$. In this way, the so-called {\em ${\bm \Lambda}$-partition of system nodes} is presented as follows:
$$[1,n]_{\mathbb{N}}:=\Xi^{uf}_{\bm \Lambda} \cup \Xi^{f}_{\bm \Lambda}.$$

\textcolor[rgb]{1,0,0}{Since that the states of nodes in set $\Xi_{{\bm \Lambda}}^{f}$ should be fixed, assuming that $\Xi_{{\bm \Lambda}}^{f}:=\{\omega_1,\omega_2,\ldots,\omega_s\}$, the states of these nodes should be globally stabilized at $(\alpha_{\omega_1},\alpha_{\omega_2},\ldots,\alpha_{\omega_s})$ with $\alpha_{\omega_j}=r$ if and only if there exists a state $\gamma\in{\bm \Lambda}$ such that $\gamma_k=r$, $\forall r\in\mathscr{D}$.}

\subsection{Determining Pinned Nodes}
\textcolor[rgb]{1,0,0}{Having defined the ${\bm \Lambda}$-partition of system nodes, we now begin the design of distributed pinning controllers, which globally stabilize BN (\ref{equ-BN}) towards the given set ${\bm \Lambda}$.}

\textcolor[rgb]{1,0,0}{Here, we first sketch the main ideas behind this stabilizing controllers. The overall process is generally divided into the selection of pinned nodes and the design for state feedback controllers and logical couplings. To be specific, the pinned node set is composed of three parts of pinned nodes. In Part I, the pinned nodes are determined as those corresponding to the ending vertices of arcs directly connecting from set $\Xi^{uf}_{\bm \Lambda}$ to set $\Xi^{f}_{\bm \Lambda}$. The pinned nodes in the Part II are selected as those corresponding to the ending vertices of feedback arc set of the local network structure induced by set $\Xi^{f}_{\bm \Lambda}$. Finally, the pinned nodes in Phase III are picked as the remain nodes whose steady state should be adjusted. After determining the pinned nodes are determined, we design the state feedback controllers and logical couplings for the different parts of pinned nodes.}

\textbf{Part I:} Finding all the directed edges $\hat{e}_1,\hat{e}_2,\ldots,\hat{e}_{q}$ in digraph ${\bf G}$ with $\textbf{O}^{-}_{\hat{e}_i}\in \{v_j|j\in\Xi^{uf}_{\bm \Lambda}\}$ and $\textbf{O}^{+}_{\hat{e}_i}\in\{v_j|j\in\Xi^{f}_{\bm \Lambda}\}$, the Part I pinned node set $\hat{\psi}$ is determined as $\hat{\psi}:=\left\{j|v_j\in\{\textbf{O}^{+}_{\hat{e}_1}, \textbf{O}^{+}_{\hat{e}_2}, \ldots, \textbf{O}^{+}_{\hat{e}_q}\}\right\}$.

\textbf{Part II:} \textcolor[rgb]{1,0,0}{Consider the digraph $\hat{{\bf G}}:=(\hat{{\bf V}},\hat{{\bf E}})$ obtained by deleting edges $\hat{e}_1,\hat{e}_2,\ldots,\hat{e}_{q}$ from digraph ${\bf G}$, for which we define the subgraph $\hat{{\bf G}}_{1}$ induced by the node set $\Xi^{f}_{\bm \Lambda}$ as follows:}
$$ \hat{{\bf G}}_{1}=(\hat{{\bf V}}_{1},\hat{{\bf E}}_{1}):=(\hat{{\bf V}}_{1},{\bf E}\cap(\hat{{\bf V}}_{1} \times \hat{{\bf V}}_{1}))$$
with $\hat{{\bf V}}_{1}=\{v_j|j\in\Xi^{f}_{\bm \Lambda}\}$. In order to make the digraph $\hat{{\bf G}}$ be acyclic, we specially take care of the feedback arc set of digraph $\hat{{\bf G}}_1$. \textcolor[rgb]{1,0,0}{Since the minimum feedback vertex/arc set problem are both NP-hard, we only approximate the minimum feedback arc set $\{\breve{e}_{1},\breve{e}_{2},\ldots,\breve{e}_{l}\}$ of subgraph $\hat{{\bf G}}_1$ via the algorithm developed in \cite{even1998approximating}, which can be determined in time $O(ms^2)$ with $s:=|\hat{{\bf V}}_{1}|=|\Xi^{f}_{\bm \Lambda}|$ and $m=|\hat{{\bf E}}_{1}|$.} Accordingly, we select the pinned node set as the nodes corresponding to the ending vertices of edges $\breve{e}_{1},\breve{e}_{2},\ldots,\breve{e}_{l}$ as $\breve{\psi}:=\left\{j|v_j\in\{\mathbf{O}^{+}_{\breve{e}_{1}},\mathbf{O}^{+}_{\breve{e}_{2}},\dots,\mathbf{O}^{+}_{\breve{e}_{l}}\}\right\}$.

\textbf{Phase III:} Finally, the third part of pinned nodes is selected as the nodes in the set $\Xi^{f}_{\bm \Lambda}$ whose steady states of nodes should be modified. \textcolor[rgb]{1,0,0}{They are determined by checking whether or not equation (\ref{equ-fixedpointequ}) holds. If (\ref{equ-fixedpointequ}) does not hold for $j\in \Xi^f_{{\bm \Lambda}}\backslash(\breve{\psi}\cup\hat{\psi})$, node $j$ is pinned and is collected by $\vec{\psi}$:}
\begin{equation}\label{equ-fixedpointequ}
\begin{aligned}
\textcolor[rgb]{1,0,0}{\zeta}({\alpha}_{j})&=A_{j}\left(\ltimes_{i\in {\bf N}_{\textcolor[rgb]{1,0,0}{j}}}\textcolor[rgb]{1,0,0}{\zeta}({\alpha}_{i})\right),~j\in\textcolor[rgb]{1,0,0}{\Xi^f_{{\bm \Lambda}}}\backslash(\breve{\psi}\cup\hat{\psi}).
\end{aligned}
\end{equation}

Thus, the pinned nodes are totally collected by $\psi:=\breve{\psi}\cup\hat{\psi}\cup\vec{\psi}$.

\subsection{Designing Feedback Controllers and Logical Couplings}
In what follows, we shall devote to designing the state feedback controllers $\phi_j$ and logical couplings $\oplus_j$ for different pinned nodes.

As mentioned above, \textcolor[rgb]{1,0,0}{arcs $e_1,e_2,\dots,e_{q},\breve{e}_1,\breve{e}_2,\ldots,\breve{e}_l$ are desired to be removed from the network structure ${\bf G}$, wherein the directed edges ending with vertex $v_j$, $j\in \hat{\psi}\cup\breve{\psi}$ are recorded by a list} $\{ e_{\varpi_j^1,j},e_{\varpi_j^2,j},\ldots,e_{\varpi_j^{\sigma_j},j}\}$. Let lists $\tilde{\textbf{N}}^c_j:=\{ \varpi_j^1,\varpi_j^2,\ldots,\varpi_j^{\sigma_j}\}$ and $\textbf{N}_j:=\{\kappa_j^1,\kappa_j^2,\ldots,\kappa_j^{\epsilon_j}\}$. Due to $\tilde{\textbf{N}}_j^c\subseteq\textbf{N}_j$, we can assume that the orders of numbers $\varpi_j^1,\varpi_j^2,\ldots,\varpi_j^{\sigma_j}$ in the list $\textbf{N}_{j}$ are $\varsigma_j^1,\varsigma_j^2,\ldots,\varsigma_j^{\sigma_j}$, respectively. Defining matrix
\textcolor[rgb]{1,0,0}{\begin{equation}\label{equ-wj}
\begin{aligned}
W_j=&W_{[2,2^{\varsigma_j^{\sigma_j}-1}]}\ltimes W_{[2,2^{\varsigma_j^{\sigma_j-1}}]} \\
&~~~~~~~~~\ltimes W_{[2,2^{\varsigma_j^{\sigma_j-2}+1}]} \ltimes \cdots \ltimes W_{[2,2^{\varsigma_j^{1}+\sigma_j-2}]} \in\mathscr{L}_{2^{\epsilon_j}\times 2^{\epsilon_j}},
\end{aligned}
\end{equation}}
the product \textcolor[rgb]{1,0,0}{$\ltimes_{i\in{\bf N}_j}x_i(t)$ can be rewritten as}
\begin{equation}\label{equ-extract}
\ltimes_{i\in{\bf N}_j}x_i(t)=W_j \left(\ltimes_{i\in \hat{{\bf N}}^c_j}x_i(t)\right) \left(\ltimes_{i\in \hat{{\bf N}}_j}x_i(t)\right)
\end{equation}
with $\tilde{{\bf N}}_j:={\bf N}_j\backslash\tilde{{\bf N}}^c_j$.

\textcolor[rgb]{1,0,0}{With these pieces in place, we solve the feasible structure matrices $M_{\oplus_j}\in\mathscr{L}_{2 \times 4}$ and $S_{\phi_j}\in\mathscr{L}_{2 \times 2^{\mid {\bf N}_j \mid}}$ of $\oplus_j$ and $\phi_j$, respectively. Consider a pinned node $j\in\hat{\psi}\cup\breve{\psi}$.} Variables ${\bm x}_{\varpi_j^i}(t)$, $i\in[1,\sigma_j]_{\mathbb{N}}$ would not be functional in the controlled nodal dynamics $\tilde{{\bm f}}_j:={\bm u}_j\oplus_j {\bm f}_j([{\bm x}_i]_{i\in \textbf{N}_j})$, for all $j\in \hat{\psi}\cup\breve{\psi}$. Substituting (\ref{equ-extract}) into the nodal dynamics $\hat{{\bm f}}_j$ follows that
\begin{equation}\label{equ-nonvariables}
x_j(t+1)=S_{\tilde{{\bm f}}_j}\ltimes_{i\in\textbf{N}_j}x_i(t)=S_{\tilde{{\bm f}}_j}W_j(\ltimes_{i\in\hat{\textbf{N}}^{c}_j}x_i(t))(\ltimes_{i\in\hat{\textbf{N}}_j}x_i(t)).
\end{equation}

By Lemma \ref{lem-extract}, we know that the matrix $S_{\tilde{{\bm f}}_j}W_j$ in (\ref{equ-nonvariables}) can be written as $S_{\tilde{{\bm f}}_j}W_j=\tilde{A}_j(I_{2^{\mid \tilde{\textbf{N}}_j \mid}}\otimes {\bf 1}^\top_{2^{\mid \tilde{\textbf{N}}^c_{j} \mid}})$ with $\tilde{A}_j\in\vec{\mathscr{L}}^c_{2\times 2^{\mid \tilde{\textbf{N}}_j \mid}}$ \textcolor[rgb]{1,0,0}{and $\tilde{A}_j(\ltimes_{i\in\tilde{{\bf N}}_j} \alpha_i)=\alpha_j$. Thus, matrix $S_{\tilde{{\bm f}}_j}$ can be expressed as $S_{\tilde{{\bm f}}_j}=\tilde{A}_j(I_{2^{\mid \tilde{\textbf{N}}_j \mid}}\otimes {\bf 1}^\top_{2^{\mid \tilde{\textbf{N}}^c_{j} \mid}})W_j^\top$.}

On the other hand, the algebraic form of
\begin{equation}
{\bm x}_j(t+1)={\bm u}_j(t)\oplus_j {\bm f}_j([{\bm x}_i(t)]_{i\in \textbf{N}_j}),~j\in\hat{\psi}\cup\breve{\psi}
\end{equation}
can be represented as
\begin{equation*}
\begin{aligned}
x_j(t+1)&=M_{\oplus_j}\hat{u}_j(t) S_{{\bm f}_j}(\ltimes_{i\in \textbf{N}_j}x_i(t))\\
&=M_{\oplus_j}S_{\phi_j}(\ltimes_{i\in \textbf{N}_j}x_i(t)) S_{{\bm f}_j}(\ltimes_{i\in \textbf{N}_j}x_i(t))\\
&=M_{\oplus_j}S_{\phi_j} (I_{2^{\mid \textbf{N}_j \mid}}\otimes S_{{\bm f}_j}) \Phi_{2^{\mid \textbf{N}_j \mid}} (\ltimes_{i\in \textbf{N}_j}x_i(t)).
\end{aligned}
\end{equation*}
To compute the unknown structure matrices $M_{\oplus_j}$ and $S_{\phi_j}$, we can establish the following matrix equation:
\begin{equation}\label{equ-logicalequationI}
M_{\oplus_j}S_{\phi_j} (I_{2^{\mid \textbf{N}_j \mid}}\otimes S_{{\bm f}_j}) \Phi_{2^{\mid \textbf{N}_j \mid}} = \tilde{A}_j(I_{2^{\mid \textbf{N}_j \mid}}\otimes {\bf 1}^\top_{2^{\mid \textbf{N}^c_{j} \mid}})W^\top_j.
\end{equation}

Then, considering pinned node set $\vec{\psi}\backslash(\hat{\psi}\cup\hat{\psi})$. Finding matrices $\tilde{A}_{j}\in\vec{\mathscr{L}}^c_{2\times 2^{\mid \textbf{N}_j \mid}}$ satisfying (\ref{equ-fixedpointequ}), structure matrices $M_{\oplus_j}$ and $S_{\phi_j}$ for this part of pinned nodes can be derived from the following equations:
\begin{equation}\label{equ-logicalequationIII}
\begin{aligned}
M_{\oplus_j}S_{\phi_j} (I_{2^{\mid \textbf{N}_j \mid}}\otimes S_{{\bm f}_j}) \Phi_{2^{\mid \textbf{N}_j \mid}} = \tilde{A}_j,~j\in \vec{\psi}\backslash(\breve{\psi}\cup\hat{\psi}).
\end{aligned}
\end{equation}
with $M_{\oplus_{j}}\in\mathscr{L}_{2\times 4}$ and $S_{\phi_{j}}\in\vec{\mathscr{L}}^c_{2\times 2^{\mid \textbf{N}_j \mid}}$.

Finally, the set stabilizing controllers can be derived by conversely converting the structure matrices into their logical form.
\begin{theorem}
Given a set ${\bm \Lambda}\subseteq \mathscr{D}^n$, the pinning controlled BN (\ref{equ-BN-pin}) will be globally ${\bm \Lambda}$-stabilized.
\end{theorem}
\begin{proof}
\textcolor[rgb]{1,0,0}{According to Lemma \ref{lem-extract}, we know that under the distributed pinning controllers designed above the subnetwork of} BN (\ref{equ-BN-pin}) induced by set $\Xi^f_{{\bm \Lambda}}$ would not involve the nodal information from the remain nodes. Besides, the network structure of this subnetwork is acyclic.

\textcolor[rgb]{1,0,0}{Using Lemma \ref{lemma-globalconvergence} to imply that subnetwork induced by set $\Xi_{{\bm \Lambda}}^f$ is globally stable, we complete the proof of this theorem by showing that its steady state is $(\alpha_{\omega_1},\alpha_{\omega_2},\cdots,\alpha_{\omega_s})$. For $j\in\Xi_{{\bm \Lambda}}^f\backslash(\vec{\psi}\cup\breve{\psi}\cup\hat{\psi})$, equation (\ref{equ-fixedpointequ}) holds so its state is desired.} Considering node $j\in\hat{\psi}\cup\breve{\psi}$, the resulting nodal dynamics will be
\begin{equation}\label{equ-thm-1}
\begin{aligned}
x_j(t+1)&=M_{\oplus_j}u_j(t) S_{{\bm f}_j} (\ltimes_{i\in \textbf{N}_j}x_i(t))\\
&=M_{\oplus_j}S_{\phi_j}(\ltimes_{i\in \textbf{N}_j}x_i(t)) S_{{\bm f}_j} (\ltimes_{i\in \textbf{N}_j}x_i(t))\\
&=M_{\oplus_j}S_{\phi_j} (I_{2^{\mid \textbf{N}_j \mid}}\otimes S_{{\bm f}_j}) \Phi_{2^{\mid \textbf{N}_j \mid}} (\ltimes_{i\in \textbf{N}_j}x_i(t))\\
&=\tilde{A}_j(\ltimes_{i\in \textbf{N}_j}x_i(t)).
\end{aligned}
\end{equation}
\textcolor[rgb]{1,0,0}{Plugging $(\alpha_{\omega_1},\alpha_{\omega_2},\cdots,\alpha_{\omega_s})$ into the right hand side of (\ref{equ-thm-1}), according to the selection of $\tilde{A}_j$, for $j\in\hat{\psi}\cup\breve{\psi}$, it satisfies (\ref{equ-fixedpointequ}). Consider the nodes in set $\vec{\psi}\backslash(\hat{\psi}\cup\breve{\psi})$. Along the same line, we also know that after control the nodal dynamics of nodes in set $\vec{\psi}$ will satisfy (\ref{equ-fixedpointequ}). Since the states of nodes in set $\Xi^{uf}_{\bm \Lambda}$ can be arbitrary, \textcolor[rgb]{1,0,0}{as a whole,} BN (\ref{equ-BN-pin}) will be globally ${\bm \Lambda}$-stabilized.}
\end{proof}

\begin{remark}\label{rmk-timebound}
By resorting to Lemma \ref{lemma-globalconvergence}, we can further consider the constraint of stabilizing time. Given the time bound $\tau$, the Part II pinned node set $\breve{\psi}$ can be sought by adding some external nodes and delete some of their incoming arcs except for the feedback arc set such that the resulting digraph is acyclic and satisfies $\text{diam}({\bf G})\leq\tau$. \textcolor[rgb]{1,0,0}{Made above tweak, the guideline can be implemented similarly to above design procedure.}
\end{remark}

\subsection{Comparing Remarks}
\textcolor[rgb]{1,0,0}{We end up this section by comparing the controllers designed above with those in \cite{liff2018tnnls,liht2017jfi3039,liurj2017neuro142,zhongjie2019new} to show our improvement.}

With $\kappa:=\mid {\bm \Lambda} \mid$, the ${\bm \Lambda}$-partition of system nodes spends $O(n)$ time. Consider the time complexity to determine the pinned nodes. In Part I, checking the reachability from set $\Xi^{uf}_{\bm \Lambda}$ to set $\Xi^{f}_{\bm \Lambda}$ in digraph ${\bf G}$ can be realized in time $O(n^2)$. Besides, the pinned nodes in Parts II and III can be determined in time $O(ms^2)$ and $O(n^2)$, respectively, \textcolor[rgb]{1,0,0}{which is lower than $O(2^{n})$ given in \cite{liff2018tnnls,liht2017jfi3039,liurj2017neuro142}. To calculate the state feedback controllers and logical couplings, (\ref{equ-logicalequationI}) and (\ref{fig-acyclic-I-II}) can be solved in time $O(s2^\varepsilon)$. Thus, the total time complexity is $O(ms^2+n^2+s2^\varepsilon)$, which breaks the barrier $O(s2^\varepsilon)$ of ASSR approach. The sparse connection of biological networks asserted in some empirical observations (see, e.g., \cite{jeong2000nature}) supports that our controllers can be carried out in a reasonable amount of time.} Finally, using the node-to-node message leads to that our approach is in the distributed form.

\textcolor[rgb]{1,0,0}{While the distributed pinning controllers in \cite{zhusy2020framework,zhusy2021controlgraph,zhusy2021sensor} concern with controllability or observability, we compare our method with that in \cite{zhongjie2019new} which aims to globally stabilize BN (\ref{equ-BN}) at a preassigned steady state.} The global set stabilization we consider here is more general than and covers the global stabilization, \textcolor[rgb]{1,0,0}{in which case Lemma \ref{lemma-globalconvergence} cannot be directly applied so that the set stabilization cannot be dealt with by the results in \cite{zhongjie2019new}.} The relation between the stabilizing time and the diameter of acyclic network structures is revealed \textcolor[rgb]{1,0,0}{(cf. Lemma \ref{lem-within-time} and Remark \ref{rmk-timebound})}. \textcolor[rgb]{1,0,0}{Last but not least, we design the state feedback controllers and logical couplings after finishing search all parts of pinned nodes. This is different from the approach in \cite{zhongjie2019new} and make the eventual stabilizing controllers more concise.}

\section{Biological Simulation}\label{section-example}
\subsection{T-LGL Survival Signal Networks}
In this subsection, \textcolor[rgb]{1,0,0}{we shall deal with the set stabilization of} the network model of T-LGL survival signal in large granular lymphocyte leukemia \cite{pnas-TLGL-29nodes}, where the node number is $n=29$. As established in \cite{pnas-TLGL-29nodes}, the logical dynamics of this network are presented as in (\ref{fig-exa}), \textcolor[rgb]{1,0,0}{in which we only give the abbreviation of each gene here and refer the readers to \cite{pnas-TLGL-29nodes} for survey}. Accordingly, its network structure can be depicted as in Fig. \ref{fig-exa-BN}. In this example, we take care of the states of IL15, PDGF, PI3K, TPL2, and SPHK and would like to globally stabilize these five nodes to $1,1,0,0,0$, respectively. It amounts to study the global ${\bm \Lambda}$-stabilization of (\ref{fig-exa}) with ${\bm \Lambda}=\{(1,\sharp,\cdots,\sharp,1,\sharp,0,\sharp,\sharp,0,0,\sharp,\cdots,\sharp)\mid \sharp\in\mathscr{D} \}$.
\begin{equation}\label{fig-exa}
\begin{array}{ll}
\text{IL15}:{\bm x}_1^+={\bm x}_1,      &\text{SPHK}:{\bm x}_{15}^+={\bm x}_{11}\vee {\bm x}_{16},\\
\text{RAS}:{\bm x}_2^+={\bm x}_1,       &\text{S1P}:{\bm x}_{16}^+={\bm x}_{15},\\
\text{ERK}:{\bm x}_3^+={\bm x}_2,       &\text{sFas}:{\bm x}_{17}^+={\bm x}_{15},\\
\text{JAK}:{\bm x}_4^+={\bm x}_1,       &\text{Fas}:{\bm x}_{18}^+=\overline{{\bm x}_{17}} \vee (\overline{{\bm x}_1}\wedge \overline{{\bm x}_{11}}),\\
\text{IL2RBT}:{\bm x}_5^+={\bm x}_1,      &\text{DISC}:{\bm x}_{19}^+= {\bm x}_{18},\\
\text{STAT3}:{\bm x}_6^+={\bm x}_4,      &\text{Caspase}:{\bm x}_{20}^+= \overline{{\bm x}_{1}}\wedge {\bm x}_{19},\\
\text{IFNGT}:{\bm x}_7^+={\bm x}_5\vee {\bm x}_6^+,      &\text{Apoptosis}:{\bm x}_{21}^+= {\bm x}_{20},\\
\text{FasL}:{\bm x}_8^+= &\text{LCK}:{\bm x}_{22}^+= {\bm x}_{1},\\
~~~[{\bm x}_6\wedge ({\bm x}_3\vee {\bm x}_5)] \vee {\bm x}_{14},&\text{MEK}:{\bm x}_{23}^+= {\bm x}_{2},\\
\text{PDGF}:{\bm x}_9^+={\bm x}_9,&\text{GZMB}:{\bm x}_{24}^+= {\bm x}_{4},\\
\text{PDGFR}:{\bm x}_{10}^+={\bm x}_9,&\text{IL2RAT}:{\bm x}_{25}^+= {\bm x}_{12},\\
\text{PI3K}:{\bm x}_{11}^+={\bm x}_{10},&\text{FasT}:{\bm x}_{26}^+= {\bm x}_{14},\\
\text{IL2}:{\bm x}_{12}^+=\overline{{\bm x}_4 \vee {\bm x}_{11}},&\text{RANTES}:{\bm x}_{27}^+= {\bm x}_{14},\\
\text{BcIxL}:{\bm x}_{13}^+=\overline{{\bm x}_4 \vee {\bm x}_{11}},&\text{A20}:{\bm x}_{28}^+= {\bm x}_{14},\\
\text{TPL2}:{\bm x}_{14}^+={\bm x}_{11},&\text{FLIP}:{\bm x}_{29}^+= {\bm x}_{11}.
\end{array}
\end{equation}

\begin{figure}
\centering
\subfigure[Network structure of survival signal network (\ref{fig-exa}).]{
\includegraphics[scale=0.35]{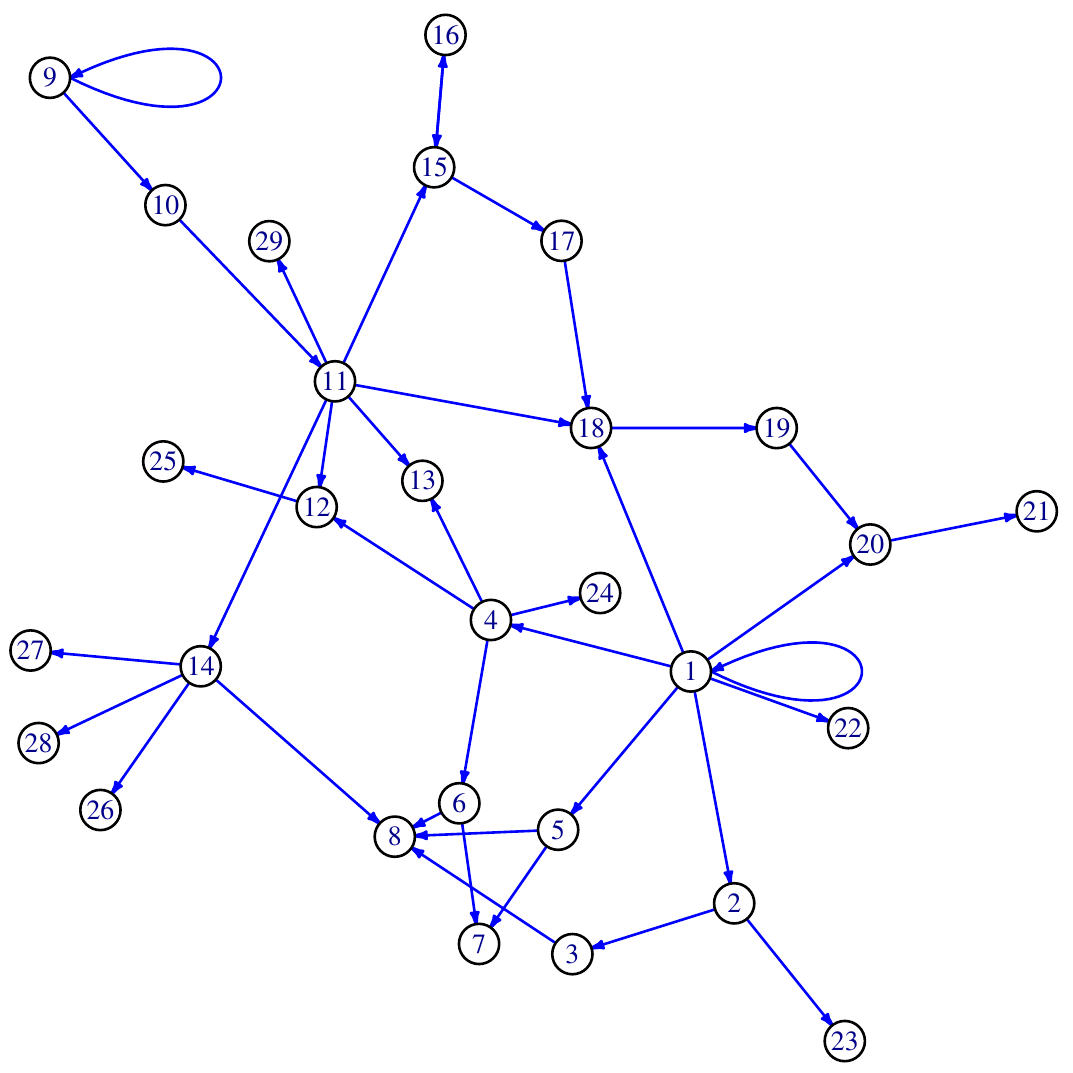}
\label{fig-exa-BN}
}
\quad
\subfigure[State transition graph of network (\ref{fig-exa}) with random $2000$ initial states.]{
\includegraphics[scale=0.35]{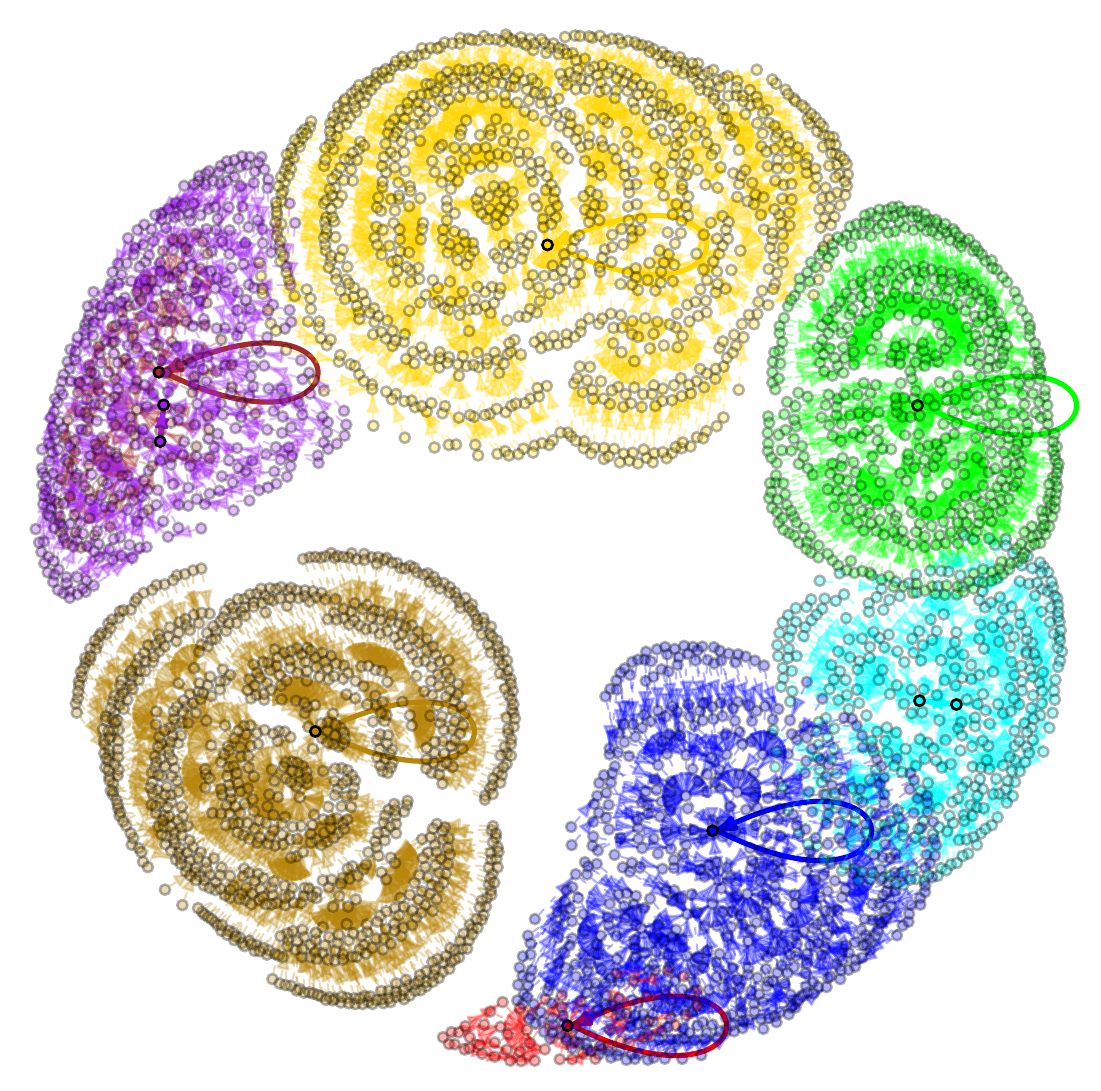}
\label{fig-exa-stg}
}
\caption{Network structure and state transition graph of BN (\ref{fig-exa}).}\label{fig-acyclic-I-II}
\end{figure}

First of all, we notice that, without any external control inputs, BN (\ref{fig-exa}) is not globally ${\bm \Lambda}$-stable. As ${\bm x}_1(t+1)={\bm x}_1(t)$, it claims that ${\bm x}_1(t)= 1$ if ${\bm x}_1(0)=1$ and ${\bm x}_1(t)= 0$ if ${\bm x}_1(0)=0$, $\forall t\in\mathbb{N}$. Thus, starting from any initial state with ${\bm x}_1(0)=0$, the state trajectory of BN (\ref{fig-exa}) will not enter the set ${\bm \Lambda}$ at any time instant. On the other hand, the state transition graph in Fig. \ref{fig-exa-stg} indicates that BN (\ref{fig-exa}) has eight attractors, some of which are not completely contained by set ${\bm \Lambda}$. In what follows, we shall be designing the distributed pinning controllers to achieve the global set stabilization of BN (\ref{fig-exa}).

Following the ${\bm \Lambda}$-partition in Subsection \ref{subsec-setpartition}, we can easily derive that $\Xi_{{\bm \Lambda}}^{f}=\{1,9,11,14,15\}$ and $\Xi^{uf}_{{\bm \Lambda}}=[1,n]_{\mathbb{N}}\backslash \Xi^{f}_{{\bm \Lambda}}$. Moreover, we can derive the states of nodes in set $\Xi_{{\bm \Lambda}}^f$ as $\alpha_{14}=\alpha_{15}=0$ and $\alpha_{1}=\alpha_{9}=\alpha_{11}=1$. From the network structure presented in Fig. \ref{fig-exa-BN}, it is noticed that the arcs from sets $\Xi_{\bm \Lambda}^{uf}$ to $\Xi^{f}_{\bm \Lambda}$ are only two edges $e_{10 \mapsto 11}$ and $e_{16 \mapsto 15}$. Thus, the Part I pinned nodes are collected by $\hat{\psi}:=\{11,15\}$. Besides, if we remove edges $e_{10 \mapsto 11}$ and $e_{16 \mapsto 15}$, the minimum feedback arc set of subgraph induced by $\Xi^f_{{\bm \Lambda}}$ is $e_{1 \mapsto 1}$ and $e_{9 \mapsto 9}$. Thus, Part II pinned node set is $\breve{\psi}:=\{1,9\}$. Since
$$ x_{14}(t)=\delta_2[1,2]x_{11}(t),$$
its next state is desired at state $(1,1,1,0,0)$. It implies that the Part III pinned node set is $\vec{\psi}:=\emptyset$. In total, the pinned node set is $\psi:=\{1,9,11,15\}$.

In what follows, the control inputs on these pinned nodes are computed. For nodes $1$, $9$ and $11$, the functional variable of their nodal dynamics is single, so the design of state feedback controllers and logical couplings can be easily derived as $\phi_1(t)=\overline{{\bm x}_1(t)}$, $\oplus_1=\vee$ and $\phi_9(t)=\overline{{\bm x}_9(t)}$, $\oplus_9=\vee$, $\phi_{11}(t)=\overline{x_{10}(t)}$ and $\oplus_{11}=\vee$. For node $15$, the structure matrix of its nodal dynamics can be given as $S_{{\bm f}_{15}}=\delta_2[1,1,1,2]$. If we plug $A_{15}=\delta_2[1,2]$, unknown matrices
$$M_{\hat{\oplus}_{15}}=\left(  \begin{array}{cccc} \alpha_1&\alpha_2&\alpha_3&\alpha_4\\ 1-\alpha_1&1-\alpha_2&1-\alpha_3&1-\alpha_4  \end{array} \right)$$
and
$$S_{\hat{\phi}_{15}}=\left(  \begin{array}{cccc} \beta_1&\beta_2&\beta_3&\beta_4\\ 1-\beta_1&1-\beta_2&1-\beta_3&1-\beta_4\end{array}\right)$$
into (\ref{equ-logicalequationI}), one has that
\begin{equation*}
\begin{array}{l}
\alpha_1\beta_1+\alpha_3(1-\beta_1)=1,~\textcolor[rgb]{1,0,0}{\alpha_1\beta_2+\alpha_3(1-\beta_2)=1,}\\
\textcolor[rgb]{1,0,0}{\alpha_1\beta_3+\alpha_3(1-\beta_3)=0,}~\alpha_2\beta_4+\alpha_4(1-\beta_4)=0.
\end{array}
\end{equation*}
One feasible solution of above equation is \textcolor[rgb]{1,0,0}{$\alpha_2=\alpha_3=\alpha_4=0$, $\alpha_1=1$, $\beta_1=\beta_2=1$, and $\beta_3=\beta_4=0$. It indicates that $\phi_{15}(t)= {\bm x}_{11}(t)$ and $\oplus_{15}=\wedge$.} To conclude, the controlled dynamics of pinned nodes are presented as follows:
\begin{equation*}
\left\{\begin{aligned}
&{\bm x}_1(t+1)= \overline{{\bm x}_1(t)} \vee {\bm x}_1(t),\\
&{\bm x}_9(t+1)= \overline{{\bm x}_9(t)} \vee {\bm x}_9(t),\\
&{\bm x}_{11}(t+1)= \textcolor[rgb]{1,0,0}{\overline{\overline{{\bm x}_{11}(t)}\vee {\bm x}_{11}(t)}} \wedge (\overline{{\bm x}_{11}(t)}\vee {\bm x}_{11}(t)),\\
&{\bm x}_{15}(t+1)={\bm x}_{11}(t)\wedge({\bm x}_{11}(t)\vee {\bm x}_{16}(t)).
\end{aligned}
\right.
\end{equation*}
For the resulting BN, its network structure and state transition graph are presented as in Figs. \ref{fig-ns-after} and \ref{fig-stg-after}, respectively. In Fig. \ref{fig-stg-after}, \textcolor[rgb]{1,0,0}{the unique attractor of the resulting BN is}
$$\textcolor[rgb]{1,0,0}{(11111111110000000110011100000)}\in{\bm \Lambda}.$$
Thus, the designed pinning controller can globally stabilize BN (\ref{fig-exa}) towards ${\bm \Lambda}$. \textcolor[rgb]{1,0,0}{Specially, the reason why we eventually attain a globally stable BN is that the original network structure in Fig. \ref{fig-exa-BN} only have two self-loops, which are both removed. Therefore, Lemma \ref{lemma-globalconvergence} indicates that the resulting BN would be globally stable.}

\textcolor[rgb]{1,0,0}{Furthermore, we proceed to consider the situation with constrained stabilizing time. Noting that the longest path of the subgraph induced by the node set $\{1,9,11,14,15\}$ is $2$, the stabilizing time of resulting BN is less than $3$. Thus, this case is not elaborated in this example.}

\begin{figure}
\centering
\subfigure[Network structure of BN (\ref{fig-exa}) under the designed pinning control.]{
\includegraphics[scale=0.35]{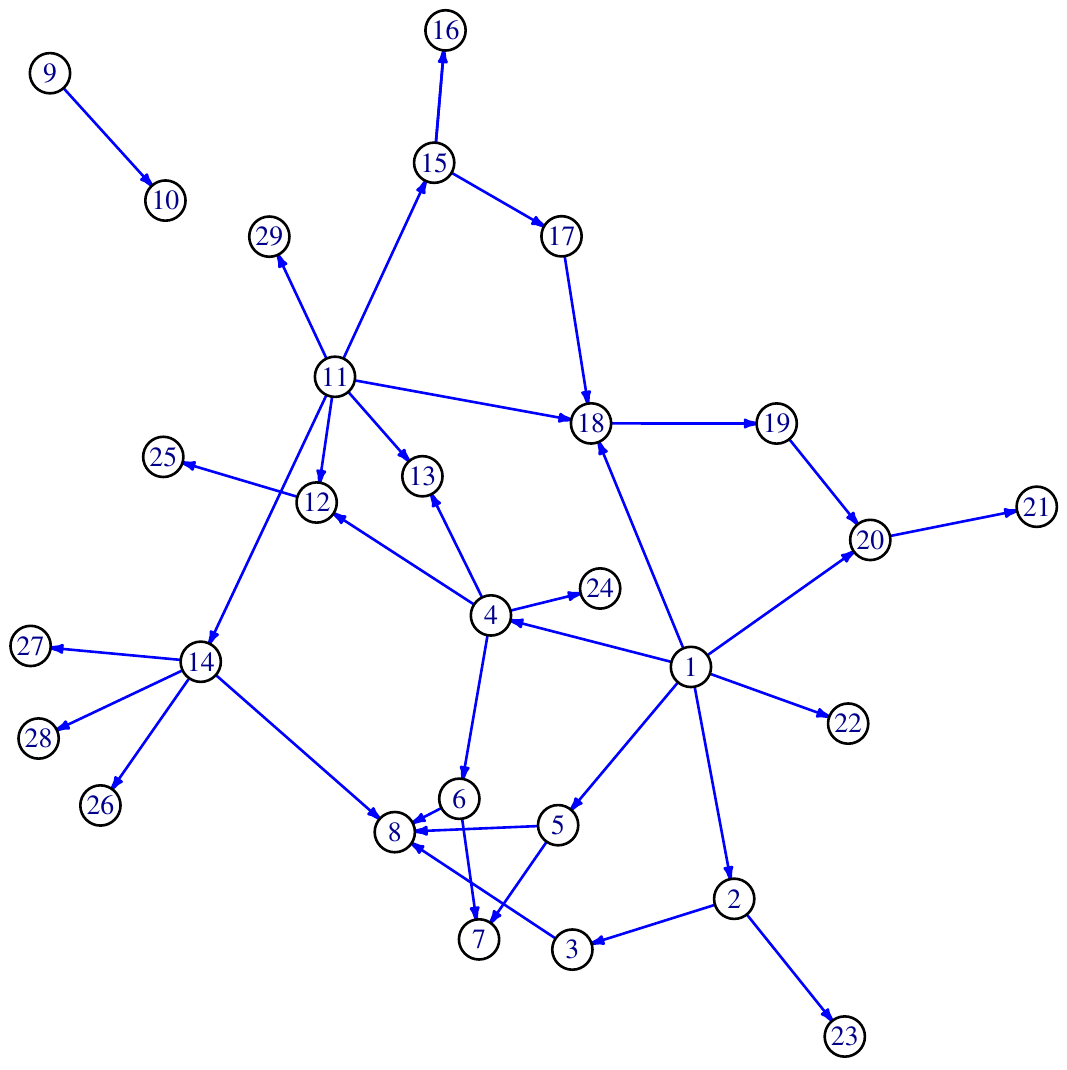}
\label{fig-ns-after}
}
\quad
\subfigure[State transition graph with random $1500$ initial states.]{
\includegraphics[scale=0.35]{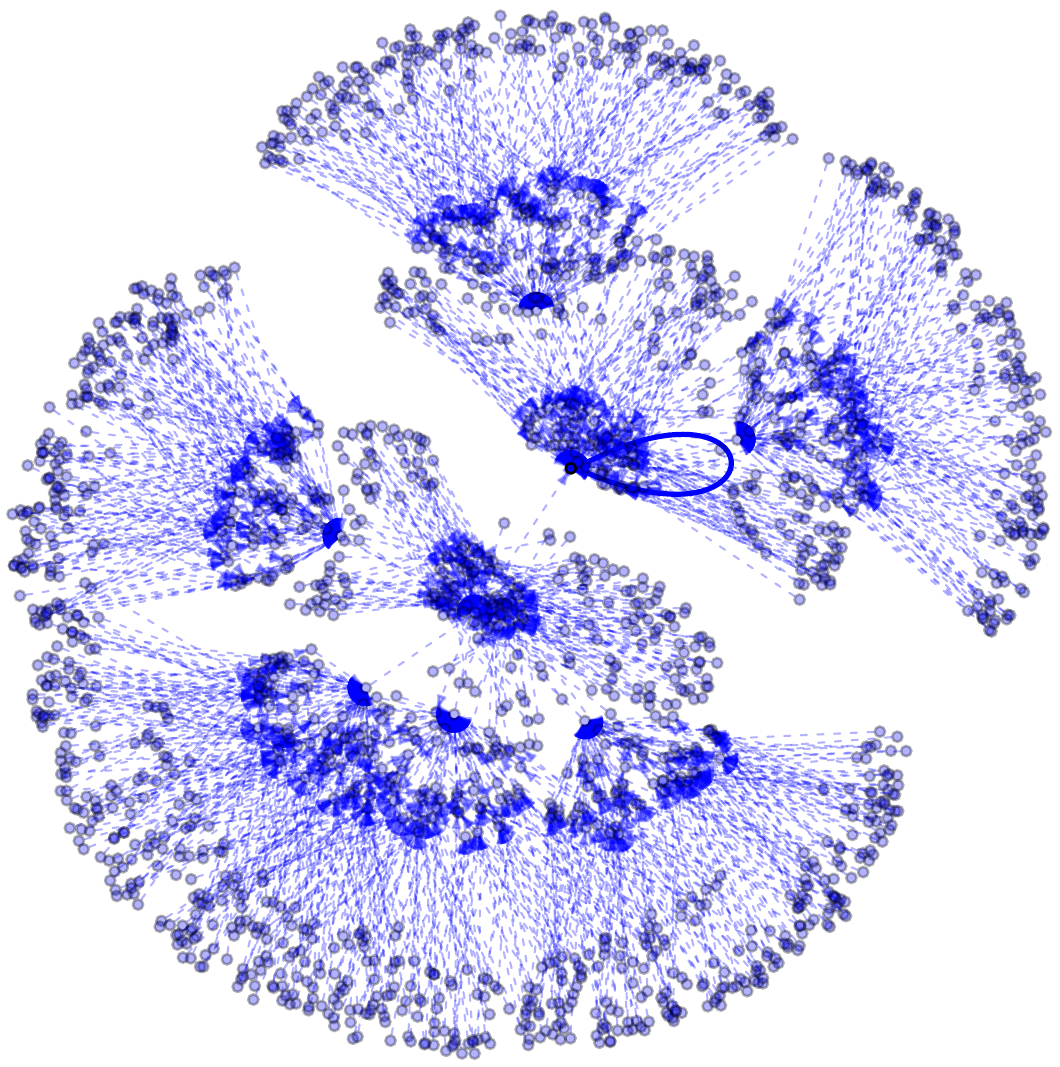}
\label{fig-stg-after}
}
\caption{Network structure and state transition graph of BN (\ref{fig-exa}) under the designed pinning controller.}\label{fig-acyclic-I-II}
\end{figure}

\subsection{\textcolor[rgb]{1,0,0}{T-Cell Receptor Signaling Networks}}
\textcolor[rgb]{1,0,0}{In this subsection, we shall turn to design the distributed pinning controller for a larger BN in cellular network, which was established in \cite{largestBN} and studied in \cite{zou2013algorithm} to describe the T-Cell Receptor Signaling Network. This model has $90$ nodes and we study the set stabilization of this network and refer the readers to \cite{largestBN} for the detailed nodal dynamics. In this example, we are interested in the states of nodes $63$, $64$, $65$, $66$, $67$, $68$, $69$, $70$ and $71$, and would like to globally stabilizing the states of all these nodes towards state $1$ within time $3$. Therefore, the desired stabilizing set here can be given as ${\bm \Lambda}=\{(\sharp,\sharp,\cdots,\sharp,1,1,1,1,1,1,1,1,1,\sharp,\sharp,\cdots,\sharp)\mid \sharp\in\mathscr{D}\}$.}

\begin{figure}[H]
\centering
\subfigure[Network structure.]{
\includegraphics[scale=0.33]{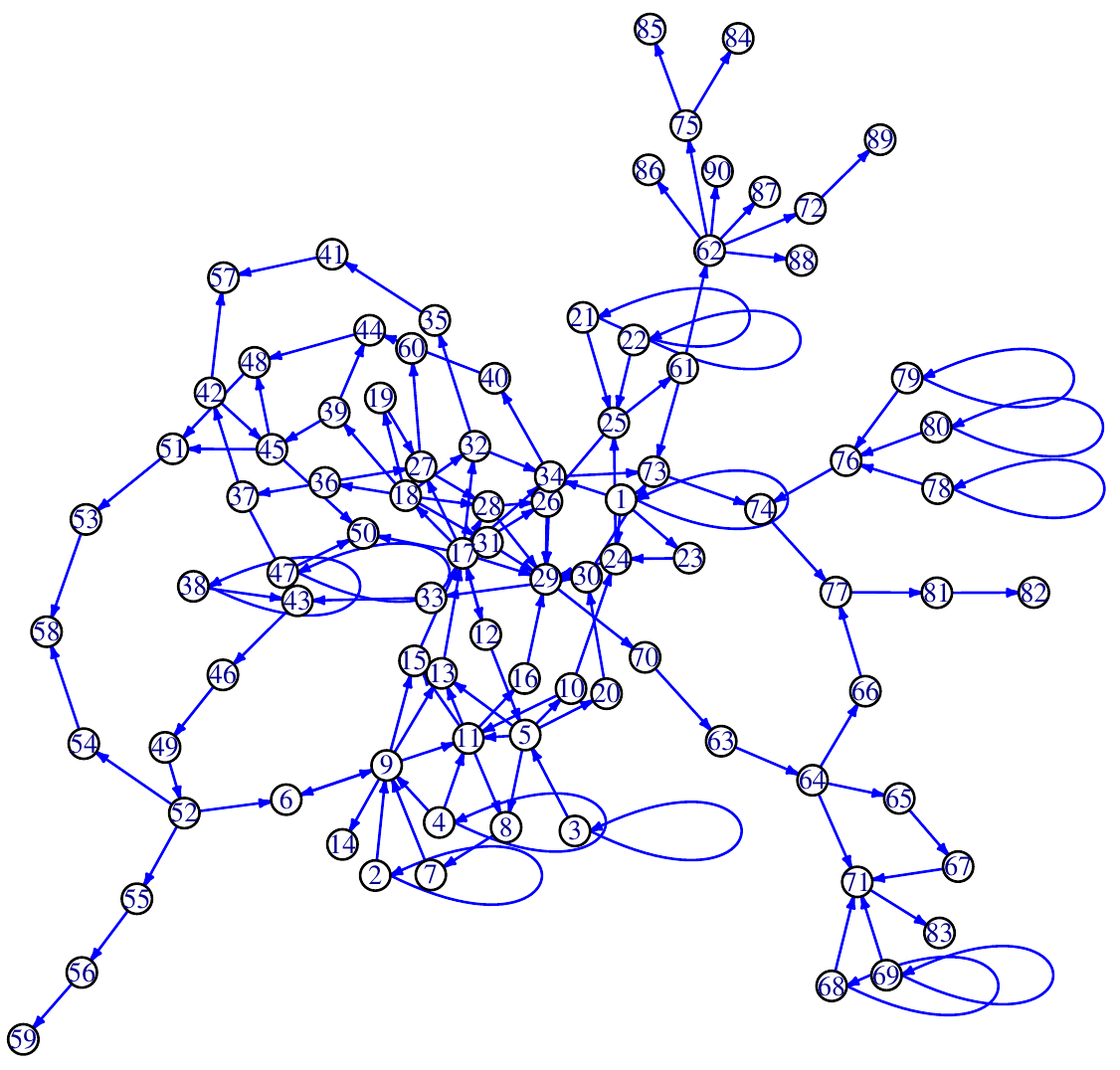}
\label{fig-ns-90nodes}
}
\quad
\subfigure[State transition graph with random $100$ initial states.]{
\includegraphics[scale=0.33]{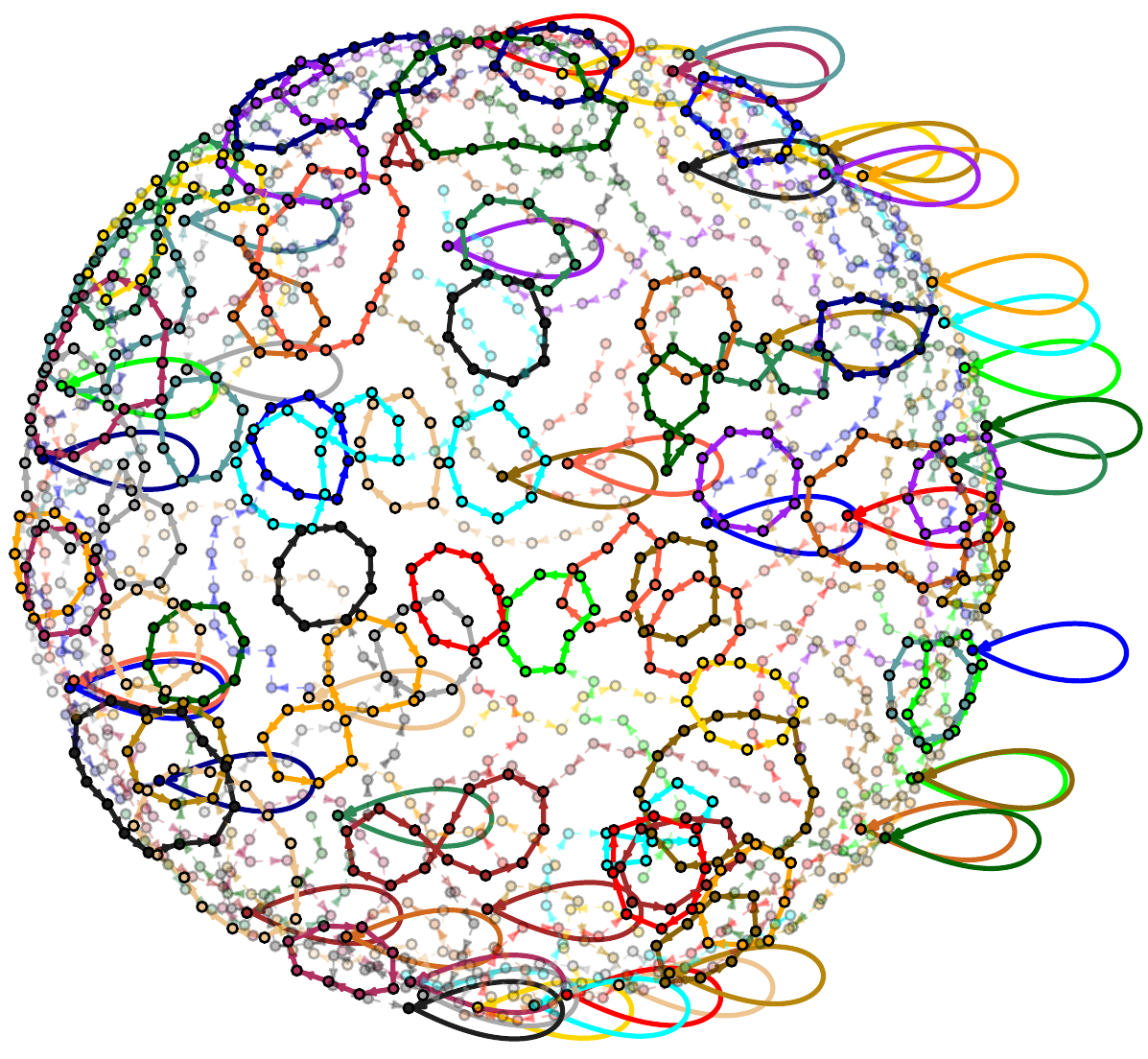}
\label{fig-stg-90nodes}
}
\caption{Network structure and state transition graph of $90$ nodes' BN.}\label{fig-90nodes}
\end{figure}

\textcolor[rgb]{1,0,0}{Given state transition graph in Fig. \ref{fig-stg-90nodes}, this BN is obviously not globally ${\bm \Lambda}$-stable. In what follows, we proceed to design the distributed pinning controller to globally stabilize BN towards set ${\bm \Lambda}$. Since all nodes $63$, $64$, $65$, $66$, $67$, $68$, $69$, $70$ and $71$ should be stabilized to $1$, one has that $\Xi_{{\bm \Lambda}}^{f}=\{63, 64, 65, 66, 67, 68, 69, 70, 71\}$.   }

\textcolor[rgb]{1,0,0}{In Part I, in order to disconnect the data flow from other nodes to set $\Xi_{{\bm \Lambda}}^{f}$, we should pin node $70$ and remove arc $e_{29\mapsto70}$. Thus, $\hat{\psi}=\{70\}$. Focusing on the subgraph induced by set $\Xi_{{\bm \Lambda}}^{f}$, there are only two self-loops $\{v_{68} \mapsto v_{68}\}$ and $\{v_{69} \mapsto v_{69}\}$. Hence, nodes $68$ and $69$ are pinned in Part II. Moreover, in the consideration of the stabilizing time being upper bounded by $3$, we also pin node $64$ here. Finally, checking whether or not equation (\ref{equ-fixedpointequ}) is satisfied. With the steady states of nodes $63$, $64$, $65$ and $66$ being desired, the steady states of nodes $67$ and $71$ should be modified, which leads to the pinned node set of this phase as $\vec{\psi}=\{67,71\}$. Hence, the total pinned node set is $\psi=\{67,68,69,70,71\}$.}

\textcolor[rgb]{1,0,0}{Since the nodal dynamics of nodes $64$, $67$, $68$, $69$ and $70$ only has one functional variable, the state feedback controller and logical coupling can be easily designed as
$$\begin{aligned}
&\textcolor[rgb]{1,0,0}{\phi_{64}=\overline{{\bm x}_{63}},\oplus_{64}=\vee,\phi_{67}={\bm x}_{65},\oplus_{67}=\vee,\phi_{68}=\overline{{\bm x}_{68}},\oplus_{68}=\vee,}\\
&\textcolor[rgb]{1,0,0}{\phi_{69}=\overline{{\bm x}_{69}},\oplus_{69}=\vee,\phi_{70}=\overline{{\bm x}_{29}},\oplus_{70}=\vee.}
\end{aligned}$$}

\textcolor[rgb]{1,0,0}{Consider node $71$, whose structure matrix can be calculated as
$$ S_{{\bm f}_{71}}=\delta_2[2,2,2,1,2,2,2,2,2,2,2,2,2,2,2,2].$$
Suppose that logical matrices $M_{\oplus_{71}}$ and $S_{\phi_{71}}$ in (\ref{equ-fixedpointequ}) as
$$M_{\oplus_{71}}=\left( \begin{array}{cccc} \alpha_1&\alpha_2&\alpha_3&\alpha_4\\1-\alpha_1&1-\alpha_2&1-\alpha_3&1-\alpha_4 \end{array} \right)$$
and
$$S_{\phi_{71}}=\left( \begin{array}{cccc} \beta_1&\beta_2&\cdots&\beta_{16}\\1-\beta_1&1-\beta_2&\cdots&1-\beta_{16} \end{array} \right).$$
Substituting above matrix $S_{{\bm f}_{71}}$ into (\ref{equ-fixedpointequ}) and selecting $\tilde{A}_{71}=\delta_2[2,2,2,1,2,2,2,2,2,2,2,2,2,2,2,2]$, we can establish the following equations:}
$$\textcolor[rgb]{1,0,0}{\begin{aligned}
&\alpha_2\beta_1+\alpha_4(1-\beta_1)=1,~\alpha_1\beta_4+\alpha_3(1-\beta_4)=0,\\
&\alpha_2\beta_i+\alpha_4(1-\beta_i)=0,~i=2,3,5,6,\cdots,16,
\end{aligned}}$$
\textcolor[rgb]{1,0,0}{of which one solution can be calculated as
$$\begin{aligned}
&\alpha_2=\alpha_3=1,~\alpha_1=\alpha_4=0,\\
&\beta_1=\beta_4=1,~\beta_2=\beta_3=\beta_5=\beta_6=\cdots=\beta_{16}=0.
\end{aligned}$$
Thus, the corresponding state feedback controller and logical coupling can be attained as $\phi_{71}={\bm x}_{64} \wedge {\bm x}_{67} \wedge {\bm x}_{69}$ and $\oplus_{71}=\bar{\vee}$.}

\textcolor[rgb]{1,0,0}{In summary, the pinning controlled BN can be presented as}
$$\begin{aligned}
&\textcolor[rgb]{1,0,0}{{\bm x}_{64}(t+1)=\overline{{\bm x}_{63}(t)}\vee {\bm x}_{63}(t),}\\
&\textcolor[rgb]{1,0,0}{{\bm x}_{67}(t+1)=x_{65}(t)\vee\overline{{\bm x}_{65}(t)},}\\
&\textcolor[rgb]{1,0,0}{{\bm x}_{68}(t+1)=\overline{{\bm x}_{68}(t)}\vee {\bm x}_{68}(t),}\\
&\textcolor[rgb]{1,0,0}{{\bm x}_{69}(t+1)=\overline{{\bm x}_{69}(t)}\vee {\bm x}_{69}(t),}\\
&\textcolor[rgb]{1,0,0}{{\bm x}_{70}(t+1)=\overline{{\bm x}_{70}(t)}\vee {\bm x}_{29}(t),}\\
&\textcolor[rgb]{1,0,0}{{\bm x}_{71}(t+1)=({\bm x}_{64}(t) \wedge {\bm x}_{67}(t) \wedge {\bm x}_{69}(t))} \\
&~~~~~~~~~~~~~~~~~~~~~~~~~~\textcolor[rgb]{1,0,0}{\bar{\vee} ({\bm x}_{64}(t) \wedge \overline{{\bm x}_{67}(t)}\wedge\overline{{\bm x}_{68}(t)}\wedge \overline{{\bm x}_{69}(t)})}.
\end{aligned}$$
\textcolor[rgb]{1,0,0}{To check the ${\bm \Lambda}$-stabilization of resulting BN, its state transition graph, specially that of subnetwork induced by set $\Xi_{\bm \Lambda}^f$, is drawn as in Fig. \ref{fig-90nodes-after}. From Fig. \ref{fig-stg-90nodes-local-after}, we can conclude that the pinning controlled BN will be globally ${\bm \Lambda}$-stabilized within time $2$.}
\begin{figure}
\centering
\subfigure[Global state transition graph.]{
\includegraphics[scale=0.33]{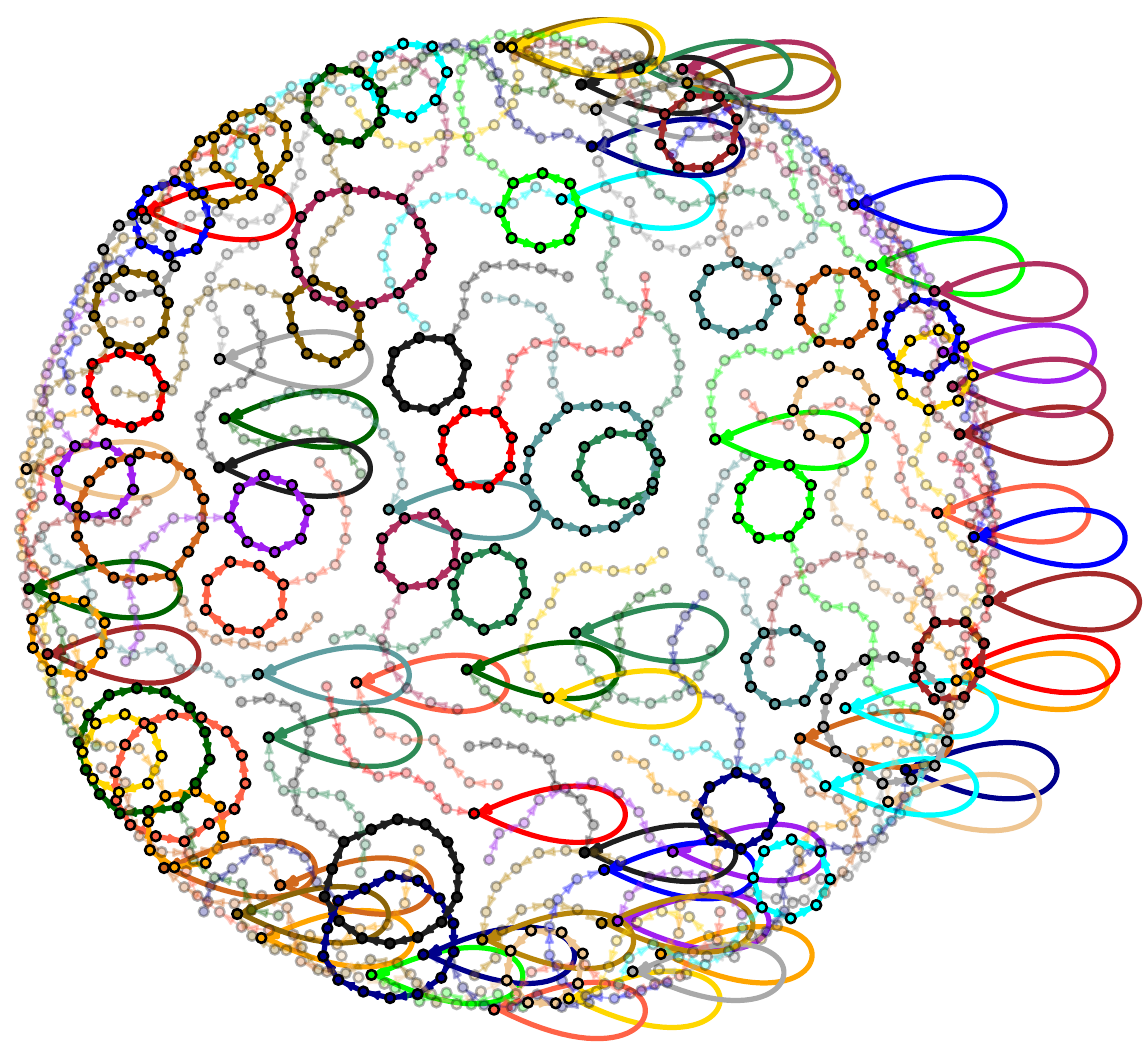}
\label{fig-stg-90nodes-global-after}
}
\quad
\subfigure[Local state transition graph.]{
\includegraphics[scale=0.33]{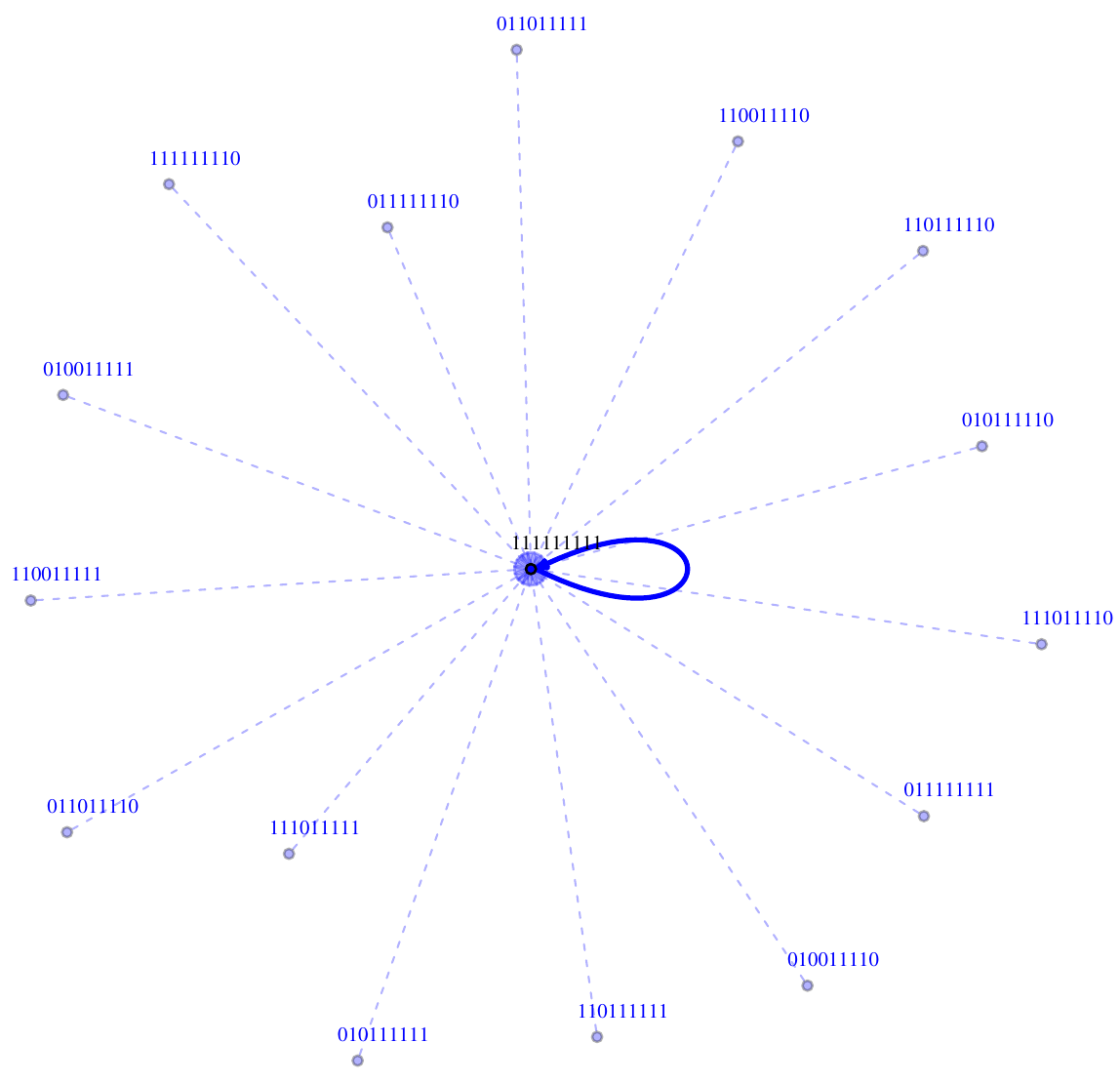}
\label{fig-stg-90nodes-local-after}
}
\caption{State transition graph of pinning controlled BN.}\label{fig-90nodes-after}
\end{figure}

\begin{remark}
\textcolor[rgb]{1,0,0}{In these two examples, compared with the traditional $L$-based pinning controller, the superiority of our method is apparent. In the first example, the overall pinned node set is $\{1,9,11,15\}$ and the largest in-degree of nodes therein is only $2$. They are respectively $\{64,67,68,69,70,71\}$ and $4$ in the second example. We only address ($2\times 4$)-dimensional and ($2\times 16$)-dimensional matrices respectively, whereas the traditional approach requires the network transition matrices of size $2^{29}\times 2^{29}=536870912 \times 536870912\gg 2\times 4$ and $2^{90}\times 2^{90}=1237940039285380274899124224 \times 1237940039285380274899124224 \gg 2\times 16$.}
\end{remark}

\section{Conclusion}\label{section-conclusion}
In this article, \textcolor[rgb]{1,0,0}{by using the global network structure and the node-to-node message change, the distributed pinning set stabilizing controllers have been designed for large-scale BNs with or without considering the stabilizing time, respectively}. After our improvement, the limitations of existing results, \textcolor[rgb]{1,0,0}{specially the high time complexity,} can be overcome to some extent. It is worthwhile to emphasize that here we pay more attention on reducing the time complexity of controllers design than pinning the less nodes. One can easily exchange thee order of selecting Part II and Part III pinned nodes to pin less nodes. However, we may pin the nodes with larger in-degree to yield a higher time complexity, which can be thought of as a cost of reducing the number of pinned nodes.

Moreover, it is stressed that in the problem of pinning control design, the synchronization \cite{lir2012tnnls840} and output regulation \cite{liht2017tac2993} of BNs cannot be solved by \textcolor[rgb]{1,0,0}{converting them into the set stabilization of another augmented BN and changing its augmented transition matrix, even if they are the variations of set stabilization. Thus, the design of pinning controllers for synchronization and output regulation is more difficult than that for set stabilization.} \textcolor[rgb]{1,0,0}{Our method presented here} is still not applicable for these two issues, \textcolor[rgb]{1,0,0}{which are left as further issues. The random version of such controllers is also interesting, where the asymptotic or finite-time behaviors of probabilistic BNs or Markovian jump BNs can be further studied.}

\end{document}